\documentclass[final,british,orivec,envcountsect,runningheads,a4paper]{llncs}

\makeatletter
\newcommand*{\installoption}[2][llncs]{
  \expandafter\newif\csname if#2\endcsname
  \@ifclasswith{#1}{#2}{
    \csname#2true\endcsname
  }{}
}
\newcommand*{\obeyoption}[2]{
  \csname if#2\endcsname\else\csname #1false\endcsname\fi
}
\newcommand*{\excludeoption}[2]{
  \csname if#2\endcsname\csname #1false\endcsname\fi
}
\makeatother

\installoption{final}
\installoption{draft}
\installoption{showframe}
\obeyoption{showframe}{draft}

\usepackage{xparse}
\usepackage{etoolbox}
\usepackage{environ}
\usepackage[final]{graphicx}
\usepackage{realboxes}

\makeatletter
\newcommand{\IfDefTF}[3]{\@ifundefined{#1}{#3}{#2}}
\newcommand{\IfDefT}[2]{\IfDefTF{#1}{#2}{}}
\newcommand{\IfDefF}[2]{\IfDefTF{#1}{}{#2}}
\newcommand{\IfLabelExistsTF}[3]{\IfDefTF{r@#1}{#3}{#2}}
\newcommand{\IfLabelExistsT}[2]{\IfLabelExistsTF{#1}{#2}{}}
\newcommand{\IfLabelExistsF}[2]{\IfLabelExistsTF{#1}{}{#2}}
\makeatother

\usepackage[T1]{fontenc}
\usepackage[utf8]{inputenc}
\usepackage[final]{microtype}
\usepackage{inconsolata}

\ifshowframe
\AtEndPreamble{%
  \usepackage{showframe}%

}
\fi

\usepackage{xcolor}

\usepackage[british]{babel}
\usepackage{csquotes}
\usepackage[all]{foreign}
\defasforeign[etsimilia]{et similia}
\defasforeign[Etsimilia]{Et similia}
\defnotforeign[wrt]{{\xperiodafter{w.r.t}}}
\defnotforeign[wl]{{\xperiodafter{w.l.o.g}}}

\usepackage{amsmath}
\usepackage{amsfonts}
\usepackage{amssymb}
\usepackage{stmaryrd}
\usepackage{centernot}
\usepackage{mathtools}
\usepackage{esvect}

\allowdisplaybreaks

\usepackage{float}
\usepackage{graphicx}
\usepackage{subcaption}
\usepackage{placeins}

\usepackage[inline]{enumitem}

\usepackage{array}
\usepackage{xtab}
\usepackage{booktabs}

\usepackage{tikz}
\usetikzlibrary{calc,arrows,positioning,intersections,decorations.pathreplacing,shapes,automata,petri,backgrounds}

\tikzset{
	place/.style={
		circle,
		thick,
		draw=blue!75,
		fill=blue!20,
		minimum size=6mm
	},
	transition/.style={
		rectangle,
		thick,
		fill=black,
		minimum width=8mm,
		inner ysep=2pt
	}
}

\usepackage{hyperref}
\usepackage{nameref}

\usepackage[capitalise,nameinlink,noabbrev]{cleveref}
\makeatletter
\newcommand{\customlabel}[4][0]{%
  \protected@write\@auxout{}{\string\newlabel{#3}{{#4}{\thepage}{#4}{#3}{}}}%
  \protected@write\@auxout{}{\string\newlabel{#3@cref}{{[#2][#1][#1]#4}{\thepage}}}%
}
\newcommand{\crefv}[1]{%
  \begingroup\@cref@compressfalse\@cref@sortfalse\cref{#1}\endgroup%
}
\newcommand{\Crefv}[1]{%
  \begingroup\@cref@compressfalse\@cref@sortfalse\Cref{#1}\endgroup%
}
\newcommand{\crefabbrev}[1]{%
  \begingroup\@cref@abbrevtrue\cref{#1}\endgroup%
}
\NewDocumentCommand\declarecrefname{s m m m m m}{
  \Crefname{#2}{#5}{#6}
  \IfBooleanTF{#1}{
    \crefname{#2}{#3}{#4}
  }{
    \if@cref@capitalise
      \crefname{#2}{#5}{#6}
    \else
      \crefname{#2}{#3}{#4}
    \fi
  }
}

\makeatother

\AtEndPreamble{\hypersetup{
  hidelinks,
  final,
}}

 \usepackage{proof-dashed}

\newcommand{\rname}[1]{%
 	\text{\scshape #1}%
}
\newcommand{\rlabel}[2]{{%
  \customlabel{rule}{#2}{#1}%
  \hypertarget{#2}{\scriptstyle #1}}}

\crefname{rule}{rule}{rules}
\Crefname{rule}{Rule}{Rules}

\spnewtheorem*{convention}{Convention}{\bfseries}{\normalfont}
\spnewtheorem*{notation}{Notation}{\bfseries}{\normalfont}
\spnewtheorem{setting}{Setting}{\bfseries}{\normalfont}

\crefname{proof}{proof of}{proofs of}
\Crefname{proof}{Proof of}{Proofs of}

\usepackage[
  sectionbib,
  square,
  numbers,
  sort&compress
  ]{natbib}

\newcommand{\defeq}{\triangleq}

\DeclareMathOperator{\dom}{dom}

\usepackage{verbatim,fancyvrb}

\newcommand{\textcode}[1]{\textnormal{\ttfamily#1}}

\newcommand{\type}[1]{\textcode{#1}\xspace}

\newcommand{\libOTM}{\textnormal{\textsf{O\kern-1ptT\kern-.5ptM}}\xspace}
\newcommand{\libSTM}{\textnormal{\textsf{S\kern-.5ptT\kern-.5ptM}}\xspace}

\newcommand{\lthrparenthesis}{(\mspace{-3.7mu}[}
\newcommand{\rthrparenthesis}{]\mspace{-3.7mu})}
\newcommand{\ctmthread}[2][t]{\lthrparenthesis#2\rthrparenthesis_{#1}}
\newcommand{\ctmthreadtr}[4]{\lthrparenthesis{#1;#2}\rthrparenthesis_{#3,#4}}

\newcommand{\ctxhole}[1][-]{[#1]}
\newcommand{\ctxP}[2][t]{\mathbb{P}_{#1}\ctxhole[#2]}

\newcommand{\ctxT}[2][t,k]{\mathbb{T}_{#1}\ctxhole[#2]}

\newcommand{\Loc}{\textsf{Loc}}
\newcommand{\Var}{\textsf{Var}}
\newcommand{\Term}{\textsf{Term}}
\newcommand{\TrName}{\textsf{TrName}}
\newcommand{\State}{\textsf{State}}

\newcommand{\ctmSt}[2]{\langle #2; #1\rangle}

\begin{document}

\mainmatter

\title{Software\;Transactional\;Memory\;with\;Interactions\thanks{Supported by Italian MIUR project PRIN 2017FTXR7S \emph{IT MATTERS} (Methods and Tools for Trustworthy Smart Systems) (M.~Miculan), and by the Independent Research Fund Denmark, Natural Sciences, grant {DFF-7014-00041} (M.~Peressotti).}}
\subtitle{Extended Version}

\author{
	Marino Miculan\inst{1}
	\and
	Marco Peressotti\inst{2}}
\authorrunning{M.~Miculan and M.~Peressotti}

\institute{
	University of Udine 
	\email{marino.miculan@uniud.it} 
	\and
	University of Southern Denmark
	\email{peressotti@imada.sdu.dk}
}

\toctitle{Software Transactional Memory with Interactions} 

\maketitle

\begin{abstract}
	Software Transactional memory (STM) is an emerging abstraction for concurrent programming alternative to lock-based synchronizations. 
	Most STM models admit only \emph{isolated} transactions, which are not adequate in multithreaded programming where transactions need to interact via shared data \emph{before} committing.
	To overcome this limitation, in this paper we present \emph{Open Transactional Memory} (OTM), a programming abstraction supporting \emph{safe, data-driven interactions} between \emph{composable} memory transactions.
	This is achieved by relaxing isolation between transactions, still ensuring atomicity.
   This model allows for \emph{loosely-coupled} interactions since transaction merging is driven only by accesses to shared data, with no need to specify participants beforehand.
\end{abstract}

\section{Introduction}
\label{sec:introduction}

Modern multicore architectures have emphasized the importance of abstractions supporting correct and scalable multi-threaded programming.   In this model, threads can collaborate by interacting on shared data structures, such as tables, message queues, buffers, etc., whose access must be regulated to avoid inconsistencies.
Traditional lock-based mechanisms (like semaphores and monitors) 
are notoriously difficult and error-prone, as they easily lead to deadlocks, race conditions and priority inversions; moreover, they are not composable and hinder parallelism, thus reducing efficiency and scalability.
\emph{Transactional memory} (TM) has emerged as a promising abstraction to replace locks \cite{moss:transactionalmemorybook,st:dc1997}.
The basic idea is to mark blocks of code as \emph{atomic}; then, execution of each block will appear either as if it was executed sequentially and instantaneously at some unique point in time, or, if aborted, as if it did not execute at all.
This is obtained by means of \emph{optimistic} executions: these blocks are allowed to run concurrently, and eventually if an interference is detected a block is automatically restarted after that its effects are rolled back.
Thus, each transaction can be viewed in isolation as a \emph{single-threaded} computation, significantly reducing the programmer's burden. Moreover, transactions are composable and ensure absence of deadlocks and  priority inversions, automatic roll-back on exceptions, and increased concurrency. 

However, in multi-threaded programming transactions may need to interact and exchange data \emph{before} committing. 
In this situation, transaction isolation is a severe shortcoming.  A simple example is a request-response interaction between two transactions via a shared buffer.  We could try to synchronize the threads accessing the buffer \type{b} by means of two semaphores \verb|c1|, \verb|c2| as follows: 
\\[\abovedisplayshortskip]
\begin{minipage}[t]{.45\textwidth}
\begin{Verbatim}[baseline=t,fontsize=\small]
// Party1 (Master)
atomically {
  <put request in b>
  up(c1);
  <some other code; may abort>
  down(c2); // wait for answer
  <get answer from b; may abort>
}
\end{Verbatim}
\end{minipage}\hfill%
\begin{minipage}[t]{.45\textwidth}
\begin{Verbatim}[baseline=t,fontsize=\small]
// Party2 (Worker)
atomically {
  down(c1); // wait for data
  <get request from b>
  <compute answer; may abort>
  <put answer in b>
  up(c2); 
}
\end{Verbatim}
\end{minipage}
\\[\belowdisplayshortskip]
Unfortunately, this solution does not work: any admissible execution requires an interleaved scheduling between the two transactions, thus violating isolation; hence, the transactions deadlock as none of them can progress.  It is important to notice that this deadlock arises because interaction occurs between threads of \emph{different} transactions; 
in fact, the solution above is perfectly fine for threads outside transactions or within the same transaction.

To overcome this limitation, in this paper we propose a programming model for \emph{safe, data-driven} interactions between memory transactions.  The key observation is that \emph{atomicity} and \emph{isolation} are two disjoint computational aspects:
\begin{itemize}
	\item an \emph{atomic non-isolated} block is executed ``all-or-nothing'', but its execution can overlap others' and \emph{uncontrolled} access to shared data is allowed;
	\item a \emph{non-atomic isolated} block is executed ``as if it were the only one'' (i.e., in mutual exclusion with others), but no rollback on errors is provided.
\end{itemize} 
Thus, a ``normal'' block of code is neither atomic nor isolated; a mutex block (like Java \emph{synchronized} methods) is isolated but not atomic; and a usual STM transaction is a block which is both atomic and isolated.  Our claim is that \emph{atomic non-isolated blocks can be fruitfully used for implementing safe composable interacting memory transactions}---henceforth called \emph{open transactions}.

In this model, a transaction is composed by several threads, called \emph{participants}, which can cooperate on shared data.  A transaction commits when all its participants commit, and aborts if any thread aborts.  Threads participating to different transactions can access to shared data, but when this happens the transactions are \emph{transparently merged} into a single one.  For instance, the two transactions of the synchronization example above would automatically merge becoming the same transaction, so that the two threads can synchronize and proceed.  Thus, this model relaxes the isolation requirement still guaranteeing atomicity and consistency; moreover, it allows for \emph{loosely-coupled} interactions since transaction merging is driven only by run-time accesses to shared data, without any explicit coordination among the participants beforehand.

\paragraph{Related work.}
Many authors have proposed mechanisms to allow transactions to interact. 
Perhaps the work closest to ours are \emph{transaction communicators} (TC) \cite{lv:ppopp11}. A transaction communicator is a (Java) object which can be accessed simultaneously by many transactions.  To guarantee consistency, dependencies between transactions are maintained at runtime: a transaction can commit only when every transactions it depends on also commit.  When dependencies form a cycle, the involved transactions must either all commit or abort together.
This differs from OTM approach, where cooperating transactions are dynamically merged and hence the dependency graph is always acyclic; thus, OTM is opaque whereas TC is not.
Other differences between TC and OTM are that our model has a formal semantics and that it can be implemented without changing neither the compiler nor the runtime (albeit it may be not very efficient).

Other authors have proposed  \emph{events}- and \emph{message passing}-based mechanisms; we mention \emph{transactional events} (TE) \cite{df:jfp08}, which are specialized to the composition of send/receive operations to simplify synchronization in communication protocols, and TIC \cite{skby:oopsla07}, where a transaction can be split into an isolated and a non-isolated transactions; this may violate local invariants and hence TIC does not satisfy opacity.
Finally,  \emph{communicating memory transactions} (CMT) \cite{lp:ppopp2011} is a model combining memory transactions with the actor model yielding \emph{transactors};
hence CMT can be seen as the message-oriented counterpart of TC.  CMT is opaque and has an efficient implementation; however
it is best suited to distributed scenarios, whereas TC and \libOTM are aimed to multi-threaded programming on shared memory---in fact, transactors can be easily implemented in \libOTM by means of queues on shared memory.  Another difference is that channel topology among transactors is established \emph{a priori}, i.e.~when the threads are created, while in \libOTM threads are created at runtime and interactions between transactions are driven by access to shared data only, whose references can be acquired at runtime.

Despite the name, our open transactions do not have much to share with \emph{open nested transactions} \cite{nietal:ppopp07}.  The latter work is about enabling physically conflicting executions of transactions, still maintaining isolation from the programmer's point of view; hence, open nested transactions cannot actually interact. 

\paragraph{Synopsis.}
In \cref{sec:monads} we present \emph{Open Transactional Memory} in the context of Concurrent Haskell.
In \cref{sec:semantics} we provide a formal operational semantics
which is used in \cref{sec:opacity} to prove that OTM satisfies the \emph{opacity} correctness criterion.
Concluding remarks and directions for future work are in \cref{sec:conclusions}.

\section{Background: Concurrency Control in Haskell}\label{sec:background}
In this paper we focus on \emph{internal concurrency}, i.e.~multiple threads in a single process cooperating through shared memory.
The dominant technique is lock-based programming
which can quickly become unmanageable as
interactions grow in complexity. In the last decade, transactional memory
has seen increasing adoption as an alternative to locks. 

In this section we briefly recall these approaches in the context of Haskell.

\subsection{Concurrent Haskell}
Haskell was born as pure lazy functional language;
side effects are handled by means of monads
\cite{pw:popl1993}.
For instance, I/O actions have type \type{IO\;a} and can be combined 
together by the monadic bind combinator \textcode{>>=}.
Therefore, the function
\textcode{putChar :: Char -> IO ()} takes a character
and delivers an I/O action that, when performed (even multiple times),
prints the given character.
Besides external inputs/outputs, values of \type{IO}
include operations with side effects on mutable (typed) cells.
A cell holding values of type \type{a}
has type \type{IORef\;a} and may be dealt with only via the following operations:
\begin{Verbatim}[fontsize=\small, tabsize=2, xleftmargin=2ex, gobble=1]
	newIORef   :: a -> IO (IORef a)
	readIORef  :: IORef a -> IO a
	writeIORef :: IORef a -> a -> IO ()
\end{Verbatim}
Concurrent Haskell \cite{pgf:popl1996}
adds support to threads which independently
perform a given I/O action as explained by
the type of the thread creation function:
\begin{Verbatim}[fontsize=\small, tabsize=2, xleftmargin=2ex, gobble=1]
	forkIO :: IO () -> IO ThreadId
\end{Verbatim}
The main mechanism for safe thread communication and synchronisation
are \emph{MVars}. A value of type \type{MVar\;a} is mutable location
(as for \type{IORef\;a}) that is either empty or full with a value of 
type \type{a}. There are two fundamental primitives to interact
with MVars:
\begin{Verbatim}[fontsize=\small, tabsize=2, xleftmargin=2ex, gobble=1]
	takeMVar :: Mvar a -> IO a
	putMvar  :: Mvar a -> a -> IO ()
\end{Verbatim}
The first empties a full location and blocks otherwise
whereas the second fills an empty location and blocks otherwise.
Therefore, MVars can be seen as one-place channels and
the particular case of \type{MVar\;()} corresponds to binary semaphores.

We refer the reader to \cite{jones:2010awkward-squad} for an 
introduction to concurrency, I/O, exceptions, and
cross language interfacing (the ``awkward squad''). 

\subsection{STM Haskell}
STM Haskell \cite{hmpm:ppopp2005} builds on Concurrent Haskell
adding \emph{transactional actions} and a transactional memory for safe 
thread communication, called \emph{transactional variables} or \emph{TVars} for short.

Transactional actions have type \type{STM\;a}
and are concatenated using \type{STM} monadic ``bind'' combinator, akin I/O actions.
A transactional action remains tentative during its execution;
(its effect) is exposed to the rest of the system by
\begin{Verbatim}[fontsize=\small, tabsize=2, xleftmargin=2ex, gobble=1]
	atomically :: STM a -> IO a
\end{Verbatim}
which takes an STM action and delivers an I/O action that,
when performed, runs the transaction guaranteeing
atomicity and isolation with respect to the rest of the
system.

Transactional variables have type \type{TVar\;a} where
\type{a} is the type of the value held and, like IOrefs,
are manipulated via the interface:
\begin{Verbatim}[fontsize=\small, tabsize=2, xleftmargin=2ex, gobble=1]
	newTVar   :: a -> STM (TVar a)
	readTVar  :: TVar a -> STM a
	writeTVar :: TVar a -> a -> STM ()
\end{Verbatim}
For instance, the following code uses monadic bind
to combine a read and write operation on a transactional
variable and define a ``transactional update'':
\begin{Verbatim}[fontsize=\small, tabsize=2, xleftmargin=2ex, gobble=1]
	modifyTVar :: TVar a -> (a -> a) -> STM ()
	modifyTVar var f = do
		x <- readTVar var
		writeOTVar var (f x) 
\end{Verbatim}
Then, \textcode{atomically\;(modifyTVar\;x\;f)} delivers 
an I/O action that applies \textcode{f} to 
the value held by \textcode{x} and updates \textcode{x}
accordingly---the two steps being executed as a single atomic isolated operation.

The primitives recalled so far cover memory interaction,
but STM allows also for \emph{composable blocking}. 
In STM Haskell, blocking translates in 
``this thread has been scheduled too early, i.e., the right conditions are not fulfilled (yet)''. The programmer can tell the scheduler about this fact by means of 
the primitive: 
\begin{Verbatim}[fontsize=\small, tabsize=2, xleftmargin=2ex, gobble=1]
	retry :: STM a
\end{Verbatim}
The semantics of \textcode{retry} is to abort the transaction 
and re-run it after at least one of the transactional variables
it has read from has been updated---there is no point in
blindly restarting a transaction. 
As showed in \cite{hmpm:ppopp2005}, this primitive suffices
to implement MVars using STM Haskell:
\begin{Verbatim}[fontsize=\small, tabsize=2, xleftmargin=2ex, gobble=1]
	data MVar a = TVar (Maybe a)
	takeMVar v = do
		m <- readTvar v
		case m of
			Nothing -> retry
			Just r -> writeTVar v Nothing >> return r
\end{Verbatim}
Thus, a value of type \type{MVar\;a} is a 
transactional variable holding a value of type \type{Maybe\;a}, 
i.e.~a value which is either \textcode{Nothing} or actually
something of type \type{a}.
A thread applying \textcode{takeMVar} to an empty
MVar is effectively blocked since it retries the transaction upon reading
\textcode{Nothing} and then it is not rescheduled until the content
of the transactional variable changes.

Finally, transactions can be composed as alternatives by means of
\begin{Verbatim}[fontsize=\small, tabsize=2, xleftmargin=2ex, gobble=1]
	orElse :: STM a -> STM a -> STM a
\end{Verbatim}
which evaluates its first argument, and if this results is a \textcode{retry} the second argument is evaluated discarding any effect of the first.

\section{Haskell interface for Isolated and Open transactions}
\label{sec:monads}

\begin{figure}[t]
	\centering
	\definecolor{sbase03}{HTML}{002B36}
	\definecolor{sbase02}{HTML}{073642}
	\definecolor{sbase01}{HTML}{586E75}
	\definecolor{sbase00}{HTML}{657B83}
	\definecolor{sbase0}{HTML}{839496}
	\definecolor{sbase1}{HTML}{93A1A1}
	\definecolor{sbase2}{HTML}{EEE8D5}
	\definecolor{sbase3}{HTML}{FDF6E3}
	\definecolor{syellow}{HTML}{B58900}
	\definecolor{sorange}{HTML}{CB4B16}
	\definecolor{sred}{HTML}{DC322F}
	\definecolor{smagenta}{HTML}{D33682}
	\definecolor{sviolet}{HTML}{6C71C4}
	\definecolor{sblue}{HTML}{268BD2}
	\definecolor{scyan}{HTML}{2AA198}
	\definecolor{sgreen}{HTML}{859900}
	\def\c#1{\textcolor{sbase02}{\textit{#1}}}
	\def\h#1{\textcolor{sbase00}{\textit{#1}}}
\begin{Verbatim}[fontsize=\small,commandchars=\\\{\}]
\h{-- Types for transactional actions ------------------------------------------}
data ITM a  \c{-- isolated atomic transactional action, return a value of type a}
data OTM a  \c{-- open atomic transactional action, return a value of type a}

\h{-- Sequencing, do notation. Here t is a placeholder for ITM or OTM ----------}
(>>=)  :: t a -> (a -> t b) -> t b
return :: a -> t a

\h{-- Running isolated and atomic actions --------------------------------------}
atomic   :: OTM a -> IO a    \c{-- delivers the IO action for an atomic one}
isolated :: ITM a -> OTM a   \c{-- delivers the atomic action for an isolated one}
retry    :: ITM a            \c{-- retry the current transaction}
orElse   :: ITM a -> ITM a -> ITM a \c{-- fall back on the second action when the}
                                    \c{-- first action issues a retry}

\h{-- Exceptions ---------------------------------------------------------------}
throw :: Exception e => e -> t a
catch :: Exception e => t a -> (e -> t a) -> t a

\h{-- Threading ----------------------------------------------------------------}
fork :: OTM () -> OTM ThreadId

\h{-- Transactional shared memory ----------------------------------------------}
data OTVar a             \c{-- sharable memory location holding values of type a}
newOTVar     :: a -> ITM (OTVar a)      
readOTVar    :: OTVar a -> ITM a        
writeOTVar   :: OTVar a -> a -> ITM ()  
\end{Verbatim}
	\caption{The base interface of \libOTM.}
	\label{fig:base-interface}
\end{figure}

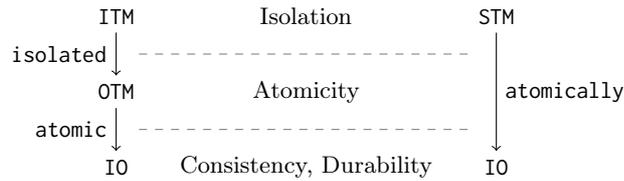
\begin{figure}[t]
	\centering
	\begin{tikzpicture}[font=\footnotesize]
	\node[] (n00) at (0,0) {\textcode{IO}};
	\node[] (n01) at (0,1) {\textcode{OTM}};
	\node[] (n02) at (0,2) {\textcode{ITM}};
	
	\draw[->] (n02) to node[left] {\textcode{isolated}} (n01);
	\draw[->] (n01) to node[left] {\textcode{atomic}} (n00);
	
	\node[] (n20) at (5,0) {\textcode{IO}};
	\node[] (n21) at (5,2) {\textcode{STM}};
	
	\draw[->] (n21) to node[right] {\textcode{atomically}} (n20);
	
	\node[] (n10) at (2.5,0) {Consistency, Durability};
	\node[] (n11) at (2.5,1) {Atomicity};
	\node[] (n12) at (2.5,2) {Isolation};
	
	\draw[gray,dashed] (.3,.5) -- (4.7,.5);
	\draw[gray,dashed] (.3,1.5) -- (4.7,1.5);
	\end{tikzpicture}
	\caption{ACID computations: splitting \textcode{atomically}.}
	\label{fig:acid-spectrum}
\end{figure}

In this section we give a brief overview of the interface for open transactions for Haskell.
In fact, OTM can be implemented in any programming language, provided we have some means to forbid irreversible effects inside transactions; we have chosen Haskell because its typing system allows us to implement this restriction quite easily.
Namely, we define two monads \type{OTM} and \type{ITM} (see  \cref{fig:base-interface}), representing the computational aspects of atomic \emph{multi-threaded open} (i.e., non-isolated) transactions and atomic \emph{single-threaded isolated} transactions, respectively. 
Transactional memory locations are values of type \textcode{OTVar} and can be manipulated by isolated transactional actions only. 

Using the construct \type{atomic}, programs in the \type{OTM} monad are executed ``all-or-nothing'' but without isolation; hence these transactions can merge at runtime. When needed, actions inside transactions can be executed in isolation by using the construct \textcode{isolated}.
Both \type{OTM} and \type{ITM} transactions are \emph{composable}; 
we exploit Haskell type system to prevent irreversible effects inside these monads.
\libOTM is a conservative extension (in fact, a drop-in replacement) 
of \libSTM \citep{hmpm:ppopp2005}; in fact, \libSTM's \type{atomically} is precisely the composition of \type{atomic} and \type{isolated} (\cref{fig:acid-spectrum}).
This allows programmers to decide the granularity of interactions; e.g., the snippet below combines read and write actions to define an isolated atomic update of a transactional location.
\begin{Verbatim}[fontsize=\small, tabsize=2, xleftmargin=2ex, gobble=1]
	modifyOTVar :: OTVar a -> (a -> a) -> ITM ()
	modifyOTVar var f = do
		x <- readOTVar var
		writeOTVar var (f x) 
\end{Verbatim}
Invariants on transactional locations can be easily checked by composing reads with checks that issue a retry if the invariant is not met, as in the snippet below.
\begin{Verbatim}[fontsize=\small, tabsize=2, xleftmargin=2ex, gobble=1]
	assertOTVar :: OTVar a -> (a -> Bool) -> ITM ()
	assertOTVar var p = do
		x <- readOTVar var
		check (p x)
	
	check :: Bool -> ITM ()
	check b = if b then return () else retry
\end{Verbatim}
By sharing \textcode{OTVar}s, non-isolated actions can share their view of transactional memory and affect each other. Consistency is guaranteed by merging transactions upon interaction thus the merged transaction may commit only if all participants agree on the final state of shared \textcode{OTVar}s.

\section{Additional examples}
\label{sec:examples}

\paragraph{Semaphores}
A semaphore is a counter with two fundamental operation:
\textcode{up} which increments the counter and
\textcode{down} which decrements the counter if
it is not zero and blocks otherwise.
Semaphores are implemented using \libOTM as OTVars
holding a counter:
\begin{Verbatim}[fontsize=\small, tabsize=2, xleftmargin=2ex, gobble=1]
	type Semaphore = OTVar Int
\end{Verbatim}
Then, \textcode{up} and \textcode{down} are
two trivial atomic and isolated updates,
with the latter being guarded by a pre-condition:
\begin{Verbatim}[fontsize=\small, tabsize=2, xleftmargin=2ex, gobble=1]
	up :: Semaphore -> ITM ()
	up s = modifyOTvar s (1+)
	
	down :: Semaphore -> ITM ()
	down s = do
		assertOTVar s (> 0)
		modifyOTVar s (-1+)
\end{Verbatim}

Actions can also be composed as alternatives
by means of the primitive \textcode{orElse}.
For instance, the following takes a family of semaphores
and delivers an action that decrements one of them, blocking
only if none can be decremented:
\begin{Verbatim}[fontsize=\small, tabsize=2, xleftmargin=2ex, gobble=1]
	downAny :: [Sempahore] -> ITM ()
	downAny (x:xs) = down x `orElse` downAny xs
	downAny [] = retry
\end{Verbatim}

\paragraph{Synchronisation}

Let us see open transactions in action by implementing a synchronisation scenario as described in \cref{sec:introduction}.
In this example a master process outsources part of an atomic computation to some thread chosen from a worker pool; data is exchanged via some shared variable, whose access is coordinated by a pair of semaphores. Notably, both the master and the worker can abort the computation at any time, leading the other party to abort as well. 
This can be achieved straightforwardly using \libOTM:
\\[\abovedisplayskip]
\begin{minipage}[t]{.5\textwidth}
	\begin{Verbatim}[fontsize=\small, tabsize=2, xleftmargin=2ex, gobble=2]
		master c1 c2 = do
			-- put request
			isolated (up c1)
			-- do something else
			isolated (down c2)
			-- get answer
	\end{Verbatim}
	\end{minipage}\begin{minipage}[t]{.5\textwidth}
	\begin{Verbatim}[fontsize=\small, tabsize=2, gobble=2]
		worker c1 c2 = do
			-- do something
			isolated (down c1)
			-- get request
			-- put answer
			isolated (up c2)
	\end{Verbatim}
	\end{minipage}
\\[\belowdisplayskip]
Both functions deliver atomic actions in
\textcode{OTM}, and hence are not isolated. We used semaphores 
for the sake of exposition but we could synchronize by means of more abstract
mechanisms, like barriers, channels or futures, 
which can be implemented using 
\libOTM as discussed in the rest of this section. 

\paragraph{Crowdfunding}
We consider a scenario in which
one party needs to atomically acquire a given number 
of resources which are offered by a dynamic group.
For sake of exposition we rephrase the example
using the metaphor of a fundraiser's ``crowdfunding campaign'': 
the resources to be acquired are the campaign goal
and the resources are donated by a dynamically determined crowd of \emph{backers}.
The implementation is shown in \cref{fig:example-funding}.

\begin{figure}[t]
	\begin{minipage}[t]{.5\textwidth}
	\begin{Verbatim}[fontsize=\small, tabsize=2, xleftmargin=2ex, gobble=2]
		type Account = OTVar Int
		type Campaign = (Account, Int)
		
		transfer :: Account -> Account 
		            -> Int -> ITM ()
		transfer a1 a2 n = do
			withdraw a1 n
			deposit a2 n
		
		newCampaign target = do
			a <- newOTvar 0
			return (a, target)
	\end{Verbatim}
	\end{minipage}\begin{minipage}[t]{.5\textwidth}
	\begin{Verbatim}[fontsize=\small, tabsize=2, gobble=2]
		backCampaign :: Account -> Campaign 
		                -> Int -> ITM ()
		backCampaign a (a',_) k = 
			transfer a a' k
		
		commitCampaign :: Account -> Campaign 
		                  -> ITM ()
		commitCampaign a (a', t) = do
			x <- readOTVar a'
			check (x >= t)
			transfer a' a x
	\end{Verbatim}
	\end{minipage}
    \caption{Crowdfunding.}
	\label{fig:example-funding}
\end{figure}

Each participant has a bank account, i.e.~an OTVar holding an integer 
representing its balance. Accounts have two operations \textcode{deposit} 
and \textcode{withdraw} which are implemented along the lines of 
\textcode{up} and \textcode{down}, respectively;  \textcode{withdraw} blocks until the account has enough funds.
A campaign have a temporary account to store funds
before transferring them to the fundraiser that closes
the campaign; this operation blocks until the goal 
is met.
Backer participants transfer a chosen amount of funds from their
account to the campaign account, but the transfer is delayed until the campaign is closed. Notice that participants do not need to know each other to coordinate.

\paragraph{Thread barriers}
Barriers are abstractions used to coordinate groups of threads;
once reached a barrier, threads cannot cross it
until all other participants reach the barrier. 
Thread groups can be either dynamic or static, depending
on whether threads may join 
the group or not. Here we consider dynamic groups. 

Threads interact with barriers with \textcode{join} for  
joining the group associated with the barrier and with \textcode{await} for blocking 
waiting all participants before crossing.

\begin{figure}[t]
	\begin{minipage}[t]{.5\textwidth}
	\begin{Verbatim}[fontsize=\small, tabsize=2, xleftmargin=2ex, gobble=2]
		type Barrier = OTVar (Int, Int)
		
		newBarrier :: ITM Barrier
		newBarrier = newOTVar (0,0)	
		
		join :: Barrier -> ITM ()
		join b = do
			assertOTVar b nobodyWaiting
			modifyOTVar b (bimap (1+) id)

		bimap f g (a, b) = (f a, g b)
	\end{Verbatim}
	\end{minipage}\begin{minipage}[t]{.5\textwidth}
	\begin{Verbatim}[fontsize=\small, tabsize=2, gobble=2]			
		await :: Barrier -> OTM ()
		await b = do
			isolated $ modifyOTVar b 
									(bimap (-1+) (1+))
			isolated $ do
				assertOTVar b nobodyRunning
				modifyOTVar b (bimap id (-1+))
		
		nobodyRunning (r,_) = r == 0
		nobodyWaiting (_,w) = w == 0
	\end{Verbatim}
	\end{minipage}
	\caption{Thread barrier.}
	\label{fig:example-barrier}
\end{figure}

Barriers can be implemented using \libOTM in few lines as shown
in \cref{fig:example-barrier}. A barrier is composed by
a transactional variable holding a pair of counters 
tracking the number of participating threads that are waiting
or running. For sake of simplicity, we prevent new joins 
during barrier crossing.
This is enforced by the assertion guarding the counter update
performed by \textcode{join}. 
Waiting and crossing correspond to the two isolated actions composing
\textcode{await}: the first changes the state of the thread from
running to waiting and the second ensures that all threads reached the barrier
before crossing and decrementing the waiting counter.
Differently from \textcode{join}, \textcode{await} cannot
be isolated: isolation would prevent other
participants from updating their state from ``running'' to ``waiting''.

This implementation is meant as a way to coordinate
concurrent transactional actions but it may be used
to coordinate concurrent I/O actions as it is.
The latter scenario could be implemented also using
\libSTM, but in this case \textcode{await} would necessarily
be an I/O action since it cannot be an
isolated atomic action (i.e., of type \type{STM\;a})---and hence, it would not be atomic either.

\paragraph{Atomic futures}
Suppose we want to delegate some task to another thread and 
collect the result once it is ready. An intuitive way to 
achieve this is by means of \emph{futures}, i.e.~``proxy results'' 
that will be produced by the worker threads. 

\begin{figure}
	\begin{minipage}[t]{.5\textwidth}
	\begin{Verbatim}[fontsize=\small, tabsize=2, xleftmargin=2ex, gobble=2]
		type Future a = OTVar (Maybe a)
					
		getFuture :: Future a -> ITM a
		getFuture f = do
			v <- readOTVar f
			case v of
				Nothing -> retry
				Just val -> return val
	\end{Verbatim}
	\end{minipage}\begin{minipage}[t]{.5\textwidth}
	\begin{Verbatim}[fontsize=\small, tabsize=2, gobble=2]
		spawn :: OTM a -> OTM (Future a)
		spawn job = do 
			future <- newOTVar Nothing
			fork (worker future)
			return future
			where
				worker :: Future a -> OTM ()
				worker future = do 
					result <- job
					writeOTVar future (Just $! result)
	\end{Verbatim}
	\end{minipage}
	\caption{Atomic futures.}
	\label{fig:example-futures}
\end{figure}

A future can be implemented in \libOTM by a TVar
holding a value of type \type{Maybe a}: either it is
``not-ready-yet'' (\textcode{Nothing}) or it holds something 
of type \type{a}.
Future values are retrieved via \textcode{getFuture} which takes 
a future and delivers an action that blocks until the value is
ready and then produces the value.
Futures are created by \textcode{spawn} which takes a transactional
action to be performed by a forked (transactional) thread. 
The complete implementation is in \cref{fig:example-futures}.

\paragraph{Petri nets}
Petri nets are a well-known (graphical) formal model for concurrent, discrete-event dynamic systems.
A Petri net is readily implemented in \libOTM 
by representing each transition by a thread, and each place by a semaphore.
Putting and taking a token from a place correspond to 
increasing (\textcode{up}) or decreasing (\textcode{down})
its semaphore---the latter blocks if no tokens are available.
Each thread repeatedly simulates the firing of the transition it 
represents, by taking tokens from its input places and putting tokens in its output places.
These semaphore operations must be 
performed atomically but not in isolation; in fact, 
isolation would prevent transitions sharing a place to
fire concurrently. 
Using \libOTM, all this is achieved in few lines:
\begin{Verbatim}[fontsize=\small, tabsize=2, xleftmargin=2ex, gobble=1]
	type Place = Semaphore
	
	transition :: [Place] -> [Place] -> IO ThreadId
	transition inputs outputs = forkIO (forever fire)
		where 
			fire = atomic $ do
				mapM_ (isolated . down) inputs
				mapM_ (isolated . up)   outputs
\end{Verbatim}
Note that, since firing is atomic but not isolated,
the above is an implementation of \emph{true concurrent} Petri nets, which is usually more difficult to achieve than interleaving semantics.

For instance, consider the Petri net in \cref{fig:example-petri},
it is immediate to implement it as follows:
\begin{Verbatim}[fontsize=\small, tabsize=2, xleftmargin=2ex, gobble=1]
	main = do
		p1 <- atomically (newPlace 1)
		p2 <- atomically (newPlace 0)
		p3 <- atomically (newPlace 0)
		p4 <- atomically (newPlace 0)
		transition [p1] [p3, p4]
		transition [p1, p2] [p4]
\end{Verbatim}
Since $p_1$ has only one token either $t_1$ or $t_2$ fires. In fact,
if $t_2$ acquires the token it will fail to acquire the other
from $p_2$ and hence its transaction retries releasing the token and leaving it to $t_1$.

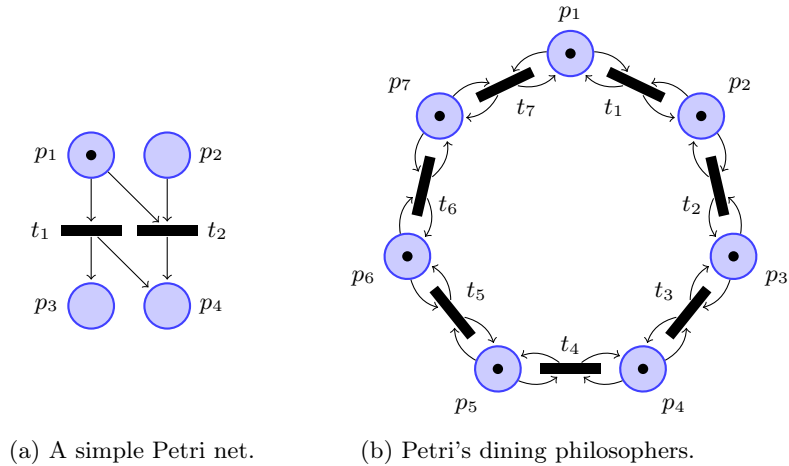
\begin{figure}[t]
	\centering
	\begin{subfigure}[b]{0.4\textwidth}
		\centering
		\begin{tikzpicture}
			\node [place,tokens=1, label=left:$p_1$] (p1) at (0,2) {};
			\node [place, label=right:$p_2$] (p2) at (1,2) {};
			
			\node [place, label=left:$p_3$] (p3) at (0,0) {};
			\node [place, label=right:$p_4$] (p4) at (1,0) {};
			
			\node [transition,label=left:$t_1$] (t1) at (0,1) {}
			edge [pre]  (p1)
			edge [post] (p3)
			edge [post] (p4);
			
			\node [transition,label=right:$t_2$] (t2) at (1,1) {}
			edge [pre]  (p1)
			edge [pre]  (p2)
			edge [post] (p4);
		\end{tikzpicture}
		\vspace{9ex}
		\caption{A simple Petri net.}
		\label{fig:example-petri}
	\end{subfigure}\quad\begin{subfigure}[b]{0.4\textwidth}
		\centering
		\begin{tikzpicture}[rotate=90]
			\def \k {7}
			\def \radius {2.2cm}
			
			\foreach \s in {1,...,\k} {
				\draw 
				let 
				\n1={-360/\k * (\s-1)}
				in 
				node [place,tokens=1,rotate=\n1,label=90:$p_\s$] 
				(p\s) at (\n1:\radius) {};
			}
			\foreach \s in {1,...,\k} {
				\draw 
				let 
				\n1={int(mod(\s,\k)+1)},
				\n2={-360/\k * (\s-1) - 360/(\k*2)}
				in 
				node [transition, rotate=\n2,label=-90:$t_\s$]
				(t\s) at 
				($(p\s)!.5!(p\n1)$) 
				{}
				edge [pre, bend right]  (p\s)
				edge [pre, bend left]  (p\n1)
				edge [post, bend left] (p\s)
				edge [post, bend right] (p\n1);
			}
		\end{tikzpicture}
		\caption{Petri's dining philosophers.}
		\label{fig:example-petri-philosophers}
	\end{subfigure}
	\caption{Examples of Petri nets.}
\end{figure}

Dijkstra's \emph{dining philosophers problem} is a textbook
classic of concurrency theory. This problem can be modelled using Petri 
nets representing each fork and philosopher as a place and as a transition
respectively; the Petri net model for the 7 philosophers instance is in  \cref{fig:example-petri-philosophers}.
Then, we can use the above implementation of Petri nets to simulate
$k$ philosophers on $k$ threads as follows:
\begin{Verbatim}[fontsize=\small, tabsize=2, xleftmargin=2ex, gobble=1]
	philosophers k = mapM_ philosopher =<< pairs
		where
			philosopher (l,r) = transition [l,r] [l,r]
			left = satomically . sequence . take k . repeat $ newPlace 1
			right = take k . drop 1 . cycle <$> left
			pairs = zip <$> right <*> left
\end{Verbatim}
Under the assumption of fair scheduling,
no execution locks.

With minor variations to \textcode{transaction}, the above 
implementation can be used to orchestrate code, using 
abstract models based on Petri nets.

\section{Formal semantics of \libOTM}
\label{sec:semantics}
In this section we provide the formal semantics of \libOTM.
Following \cite{hmpm:ppopp2005}, we fix an Haskell-like language extended with the \libOTM primitives of \cref{fig:base-interface} and characterise the behaviour of \libOTM by means of an abstract machine.

The language syntax is given by the following grammar:
\begin{alignat*}{4}
\textit{Values} &\hspace{1.2ex}& V \Coloneqq {} & r \mid \textcode{\textbackslash $x$\;->\;$M$}
		\mid \textcode{return\;$M$}
		\mid \textcode{$M$\;>>=\;$N$}
		\mid \textcode{throw\;$M$} 
		\mid \textcode{catch\;$M$\;$N$}
\\&& \mid {} & 
		\textcode{putChar\;$c$}
		\mid \textcode{getChar}
		\mid \textcode{fork\;$M$}
		\mid \textcode{atomic\;$M$}
		\mid \textcode{isolated\;$M$}
		\mid \textcode{retry}
\\&& \mid {} & 
		\textcode{$M$\;`orElse`\;$N$} 
		\mid \textcode{newOTVar\;$M$} 
		\mid \textcode{readOTVar\;$r$}
		\mid \textcode{writeOTVar\;$r$\;$M$}
\\\textit{Terms} && M \Coloneqq {} & x \mid V \mid M\,N
\end{alignat*}
where the meta-variables $x$ and $r$ range over a given countable set of variables \Var\ and of location names \Loc, respectively.
We assume Haskell typing conventions and denote the set of all well-typed terms by $\Term$. 

Terms are evaluated by an abstract state 
machine whose states are pairs $\ctmSt{\Sigma}{P}$ formed by:
\begin{itemize}
	\item
	a \emph{thread family} (or \emph{process}) $P = T_{t_1} \parallel \dots \parallel T_{t_n}$ where $t_i$ are unique \emph{thread identifiers};
	\item a \emph{memory} $\Sigma = \langle \Theta, \Delta ,\Psi \rangle$, where $\Theta : \Loc \rightharpoonup \Term$ is the \emph{heap} and $\Delta : \Loc \rightharpoonup \Term \times \TrName$ is the \emph{working memory};
	$\TrName$ is a set of names used to identify active transactions;
	$\Psi$ is a forest of threads identifiers keeping track of how threads have been forked.
\end{itemize}
Threads are the smaller unit of execution the machine scheduler operates on; they evaluate \libOTM terms and do not have any private transactional memory.
A thread $T_t$ has two forms: $\ctmthread{M}$ for threads evaluating a term $M$ outside a transaction and $\ctmthreadtr{M}{N}{t}{k}$ for threads evaluating $M$ inside transaction $k$ with continuation $N$ (the term to evaluate after that $k$ has committed).

As for traditional closed (ACID) transactions (e.g., \cite{hmpm:ppopp2005}), operations inside a transaction
are evaluated against the distributed working memory $\Delta$ and effects are propagated to the heap $\Theta$ only on commits.  
When a thread inside a transaction $k$ accesses a location outside $\Delta$ the location is \emph{claimed by transaction $k$} and remains claimed until $k$ commits, aborts or restarts. Threads in $k$ can interact only with locations claimed by $k$, but active transactions can be merged to share their claimed locations.
We denote the set of all possible states as $\State$, and reference to each projected component of $\Sigma$ by a subscript, i.e.~$\Sigma_\Theta$ for the heap and $\Sigma_\Delta$ for the working memory.
When describing  updates to the memory $\Sigma$, we adopt the convention that $\Sigma'$ has to be intended equals to $\Sigma$ except if stated otherwise, i.e.~by statements like $\Sigma'_\Theta = \Sigma_\Theta[r \mapsto M]$. Finally, $\varnothing$ denotes the empty heap and working memory.

\begin{figure}[p]
	\begin{gather*}
		\infer[\rlabel{\rname{Eval}}{rule:admn-eval}]
			{M \to V}
			{M \not\equiv V \quad \mathcal{V}[M] = V}
\\
		\infer[\rlabel{\rname{BindVal}}{rule:admn-bindv}]
		{\textcode{return\;$M$\;>>=\;$N$} \to \textcode{$N$\,$M$}}
		{}
		\qquad
		\infer[\rlabel{\rname{BindEx}}{rule:admn-binde}]
			{\textcode{e\;>>=\;$M$} \to \textcode{e}}
			{\textcode{e} \in \{\textcode{retry},\textcode{throw\;$N$}\}}
\\
		\infer[\rlabel{\rname{CatchVal}}{rule:admn-catchv}]
			{\textcode{r\;`catch`\;$M$} \to \textcode{r}}
			{\textcode{r} \in \{\textcode{retry},\textcode{return\;$N$}\}}
		\qquad
		\infer[\rlabel{\rname{CatchEx}}{rule:admn-catche}]
			{\textcode{throw\;$M$\;`catch`\;$N$} \to \textcode{$N$\,$M$}}
			{}
	\end{gather*}
	\vspace{-1ex}
	\caption{Term reductions: \small$M \to N$.}
	\label{fig:semantics-term}
%
	\begin{gather*}
		\infer[\rlabel{\rname{InChar}}{rule:input-char}]
		{\ctmSt{\Sigma}{\ctxP{\textcode{getChar}}} 
			\xrightarrow{?c}
			\ctmSt{\Sigma}{\ctxP{\textcode{return\;$c$}}}}
		{}
\\
		\infer[\rlabel{\rname{OutChar}}{rule:output-char}]
		{\ctmSt{\Sigma}{\ctxP{\textcode{putChar\;$c$}}} 
			\xrightarrow{!c}
			\ctmSt{\Sigma}{\ctxP{\textcode{return\;()}}}}
		{}
\quad
		\infer[\rlabel{\rname{TermIO}}{rule:termio}]
			{\ctmSt{\Sigma}{\ctxP{M}} 
				\xrightarrow{\tau}
				\ctmSt{\Sigma}{\ctxP{N}}}
			{M \to N}
\\
		\infer[\rlabel{\rname{ForkIO}}{rule:forkio}]
		{\ctmSt{\Sigma}{\ctxP{\textcode{fork\;$M$}}}
			\xrightarrow{\tau}
			\ctmSt{\Sigma}{\ctxP{\textcode{return\;$t'$}} \parallel 
				\ctmthread[t']{M}}}
		{t'\notin \mathsf{threads}(\ctxP{\textcode{fork\;$M$}})}
	\end{gather*}
		\vspace{-1ex}
	\caption{\type{IO} state transitions.}
	\label{fig:semantics-io}
	\begin{gather*}
		\infer[\rlabel{\rname{TermT}}{rule:term}]
			{\ctmSt{\Sigma}{\ctxT{M}} 
			\xrightarrow{\tau}
			\ctmSt{\Sigma}{\ctxT{N}}}
			{M \to N}
	\\
		\infer[\rlabel{\rname{ForkT}}{rule:fork}]
		{\ctmSt{\Sigma}{\ctxT{\textcode{fork\;$M$}}}
		 \xrightarrow{\tau}
		 \ctmSt{\Sigma'}{\ctxT{\textcode{return\;$t'$}} \parallel 
		 \ctmthreadtr{M}{\textcode{return}}{t'}{k}}}
		{t'\notin \mathsf{threads}(\ctxT{\textcode{fork\;$M$}})
		&\Sigma'_\Psi=\mathsf{add\_child}(t,t',\Sigma_\Psi)}
	\\
		\infer[\rlabel{\rname{NewVar}}{rule:newvar}]
		{\ctmSt{\Sigma}{\ctxT{\textcode{newOTVar\;$M$}}}
		 \xrightarrow{\tau}
		 \ctmSt{\Sigma'}{\ctxT{\textcode{return\;$r$}}}}
		{r \notin \dom(\Sigma_\Theta)\cup\dom(\Sigma_\Delta)
		&\Sigma'_\Delta = \Sigma_\Delta[r\mapsto (M,k)]}
	\\
		\infer[\rlabel{\rname{Read1}}{rule:read-miss}]{
			\ctmSt{\Sigma}{\ctxT{\textcode{readOTVar\;$r$}}}
			\xrightarrow{\tau}
			\ctmSt{\Sigma'}{\ctxT{\textcode{return\;$M$}}}
		}{
			r \notin \dom(\Sigma_\Delta)
			&\Sigma_\Theta(r) = M
			&\Sigma'_\Delta = \Sigma_\Delta[r \mapsto (M,k)]
		}
	\\
		\infer[\rlabel{\rname{Read2}}{rule:read-hit}]
		{\ctmSt{\Sigma}{\ctxT{\textcode{readOTVar\;$r$}}}
		 \xrightarrow{\tau}
		 \ctmSt{\Sigma'}{\ctxT[t,j]{\textcode{return\;$M$}}}}
		{\Sigma_\Delta(r) = (M,j)
		&\Sigma'_\Delta = \Sigma_\Delta[k \mapsto j]}	
	\\
		\infer[\rlabel{\rname{Write1}}{rule:write-miss}]
		{\ctmSt{\Sigma}{\ctxT{\textcode{writeOTVar\;$r$\;$M$}}}
		 \xrightarrow{\tau}
		 \ctmSt{\Sigma'}{\ctxT{\textcode{return\;()}}}}
		{r \notin \dom(\Sigma_\Delta)
		&\Sigma'_\Delta = \Sigma_\Delta[r \mapsto (M,k)]}
	\\
		\infer[\rlabel{\rname{Write2}}{rule:write-hit}]
			{\ctmSt{\Sigma}{\ctxT{\textcode{writeOTVar\;$r$\;$M$}}}
			 \xrightarrow{\tau}
			 \ctmSt{\Sigma'}{\ctxT[t,k]{\textcode{return\;()}}[k \mapsto j]}}
			{\Sigma_\Delta(r) = (N,j)
			&\Sigma'_\Delta = \Sigma_\Delta[k \mapsto j][r \mapsto (M,j)]}
	\\
		\infer[\rlabel{\rname{Or1}}{rule:orfirst}]{
			\ctmSt{\Sigma}{\ctxT{\textcode{$M$\;`orElse`\;$M'$}}}
			\xrightarrow{\tau}
			\ctmSt{\Sigma'}{\mathbb{T}_{t,j}[\textcode{op\;$N$}]}
		}{
					\textcode{op} \in \{\textcode{throw}, \textcode{return}\}
					&
					\ctmSt{\Sigma}{\ctmthreadtr{M}{\textcode{return}}{t}{k}} 
					\xrightarrow{\tau}^*
					\ctmSt{\Sigma'}{\ctmthreadtr{\textcode{op\;$N$}}{\textcode{return}}{t}{j}}
				}
	\\
		\infer[\rlabel{\rname{Or2}}{rule:orsecond}]
			{\ctmSt{\Sigma}{\ctxT{\textcode{$M$\;`orElse`\;$M'$}}}
			 \xrightarrow{\tau}
			 \ctmSt{\Sigma}{\ctxT{M'}}}
			{\ctmSt{\Sigma}{\ctmthreadtr{M}{\textcode{return}}{t}{k}} 
			 \xrightarrow{\tau}^*
			 \ctmSt{\Sigma'}{\ctmthreadtr{\textcode{retry}}{\textcode{return}}{t}{j}}}
	\\
		\infer[\rlabel{\rname{Isolated}}{rule:isolated}]
		{\ctmSt{\Sigma}{\ctxT{\textcode{isolated\;$M$}}} 
		 \xrightarrow{\tau}
		 \ctmSt{\Sigma'}{\ctxT[t,j]{\textcode{op\;$N$}}}}
		{\textcode{op} \in \{\textcode{throw}, \textcode{return}\}
		&\ctmSt{\Sigma}{\ctmthreadtr{M}{\textcode{return}}{t}{k}} 
		 \xrightarrow{\tau}^*
		 \ctmSt{\Sigma'}{\ctmthreadtr{\textcode{op\;$N$}}{\textcode{return}}{t}{j}}}
	\end{gather*}
		\vspace{-1ex}
	\caption{Transactional state transitions: \small$\ctmSt{\Sigma}{P} \xrightarrow{\tau}\ctmSt{\Sigma'}{P'}$.}
	\label{fig:semantics-tau}
\end{figure}

\begin{figure}[t]
	\begin{gather*}
		\infer[\rlabel{\rname{New}}{rule:tr-new}]
		{\ctmSt{\Sigma}{\ctmthread{\textcode{atomic\;$M$\;>>=\;$N$}}}
		 \xrightarrow{new\langle k\rangle}
		 \ctmSt{\Sigma}{\ctmthreadtr{M}{N}{t}{k}}}
		{}
	\\
		\infer[\rlabel{\rname{Commit}}{rule:tr-commit}]
		{\ctmSt{\Sigma}{\ctmthreadtr{\textcode{return\;$M$}}{N}{t}{k}}
		 \xrightarrow{co\langle k\rangle}
		 \ctmSt{\Sigma'}{\ctmthread{\textcode{return\;$M$\;>>=\;$N$}}}}
		{\Sigma'_\Theta = \mathsf{commit}(k,\Sigma)
		&\Sigma'_\Delta = \mathsf{cleanup}(k,\Sigma)}
	\\
		\infer[\rlabel{\rname{Abort1}}{rule:tr-abort-1}]
		{\ctmSt{\Sigma}{\ctmthreadtr{\textcode{throw\;$M$}}{N}{t}{k}}
		 \xrightarrow{ab\langle k, t, M\rangle}
		 \ctmSt{\Sigma'}{\ctmthread{\textcode{throw\;$M$\;>>=\;$N$}}}}
		{\Sigma'_\Theta = \mathsf{leak}(k,\Sigma)
		&\Sigma'_\Delta = \mathsf{cleanup}(k,\Sigma)
		&\Sigma'_\Psi = \mathsf{remove}(r,\Sigma_\Psi)
		&r = \mathsf{root}(t,\Sigma_\Psi)}
	\\
		\infer[\rlabel{\rname{Abort2}}{rule:tr-abort-2}]
		{\ctmSt{\Sigma}{\ctmthreadtr{M'}{N}{t'}{k}}
		 \xrightarrow{\overline{ab}\langle k, t, M\rangle}
		 \ctmSt{\Sigma'}{\ctmthread[t']{\textcode{throw\;$M$\;>>=\;$N$}}}}
		{\begingroup\def\arraystretch{1.1}\array{c}
		r = \mathsf{root}(t,\Sigma_\Psi) \qquad
		r = \mathsf{root}(t',\Sigma_\Psi) \\
		\Sigma'_\Theta = \mathsf{leak}(k,\Sigma) \qquad
		\Sigma'_\Delta = \mathsf{cleanup}(k,\Sigma) \qquad
		\Sigma'_\Psi = \mathsf{remove}(r,\Sigma_\Psi)
		\endarray\endgroup}
	\\
		\infer[\rlabel{\rname{Abort3}}{rule:tr-abort-3}]
		{\ctmSt{\Sigma}{\ctmthreadtr{M'}{N}{t'}{k}}
		 \xrightarrow{\overline{ab}\langle k, t, M\rangle}
		 \ctmSt{\Sigma'}{\ctmthread[t']{\textcode{retry}}}}
		{\begingroup\def\arraystretch{1.1}\array{c}
		r = \mathsf{root}(t,\Sigma_\Psi) \qquad
		r \neq \mathsf{root}(t',\Sigma_\Psi) \\
		\Sigma'_\Theta = \mathsf{leak}(k,\Sigma) \qquad
		\Sigma'_\Delta = \mathsf{cleanup}(k,\Sigma) \qquad
		\Sigma'_\Psi = \mathsf{remove}(r,\Sigma_\Psi)
		\endarray\endgroup}
	\\
		\infer[\rlabel{\rname{MCastAb}}{rule:tr-multicast-abort}]
		{\ctmSt{\Sigma}{P\parallel Q} \xrightarrow{ab\langle k, t, M\rangle} \ctmSt{\Sigma'}{P'\parallel Q'}}
		{\ctmSt{\Sigma}{P} \xrightarrow{ab\langle k, t, M\rangle} \ctmSt{\Sigma'}{P'}
		&\ctmSt{\Sigma}{Q} \xrightarrow{\overline{ab}\langle k, t, M\rangle} \ctmSt{\Sigma'}{Q'}}
	\\
		\infer[\rlabel{\rname{MCastCo}}{rule:tr-multicast-commit}]
		{\ctmSt{\Sigma}{P \parallel Q}  \xrightarrow{co\langle k\rangle} \ctmSt{\Sigma'}{P' \parallel Q'}}
		{\ctmSt{\Sigma}{P} \xrightarrow{co\langle k\rangle} \ctmSt{\Sigma'}{P'}
		&\ctmSt{\Sigma}{Q} \xrightarrow{co\langle k\rangle} \ctmSt{\Sigma'}{Q'}}
	\\
		\infer[\rlabel{\rname{MCastGroup}}{rule:tr-multicast-context}]
		{\ctmSt{\Sigma}{P \parallel Q} \xrightarrow{\beta} \ctmSt{\Sigma'}{P' \parallel Q}}
		{\ctmSt{\Sigma}{P} \xrightarrow{\beta} \ctmSt{\Sigma'}{P'}
		&\beta \neq \tau 
		&\mathsf{transaction}(\beta) \notin  \mathsf{transactions}(Q)}
	\end{gather*}
	\caption{Transaction management transitions: \small$\ctmSt{\Sigma}{P} \xrightarrow{\beta}\ctmSt{\Sigma'}{P'}$.}
	\label{fig:semantics-trs-mgr}
	\vspace{-2ex}
\end{figure}	

\begin{figure}
	\begin{align*}
	\mathsf{threads}(T_{t_1} \parallel \dots \parallel T_{t_n}) &\defeq \{t_1, \dots t_n\}
	\\
	\mathsf{transaction}(\beta) &\defeq k \text{ for } 
	\beta\in\{new\langle k\rangle, co\langle k\rangle, ab\langle k, t, M\rangle, \overline{ab}\langle k, t, M\rangle\}
	\\
	(\Delta[k \mapsto j])(r) &\defeq
	\begin{cases}
	\Delta(r) &\!\text{if}\ \Delta(r)= (M,l),l\neq k\\
	(M,j) &\!\text{if}\ \Delta(r) = (M, k)
	\end{cases}
	\\
	\mathsf{transactions}(P) &\defeq 
	\begin{cases}
	\mathsf{transactions}(P_1) \cup \mathsf{transactions}(P_2) & 
	\!\text{if } P = P_1 \parallel P_2\\
	\{k\}&\!\text{if } P = \ctmthreadtr{M}{N}{t}{k}\\
	\emptyset &\!\text{otherwise}
	\end{cases}
	\\
	P[k \mapsto j] &\defeq 
	\begin{cases}
	P_1[k \mapsto j] \parallel P_2[k \mapsto j] &\!\text{if } P = P_1 \parallel P_2\\
	\ctmthreadtr{M}{N}{t}{j} &\!\text{if } P = \ctmthreadtr{M}{N}{t}{k}\\
	P &\!\text{otherwise}
	\end{cases}
	\\
	\Theta[r \mapsto M](s) &\defeq
	\text{if } r = s \text{ then } M \text{ else } \Theta(s)
	\\
	\Delta[r \mapsto (M,k)](s) &\defeq
	\text{if } r = s \text{ then } (M,k) \text{ else } \Delta(s)
	\\
	\mathsf{cleanup}(k,\Sigma)(r) &\defeq 
	\text{if } \Sigma_\Delta(r) = (M,k) \text{ then } \perp \text{ else } 
	\Sigma_\Delta(r)
	\\
	\mathsf{commit}(k,\Sigma)(r) & \defeq 
	\text{if } \Sigma_\Delta(r) = (M,k) \text{ then } M \text{ else } \Sigma_\Theta(r)
	\\
	\mathsf{leak}(k,\Sigma)(r) &\defeq 
	M \!\text{ if } \Sigma_\Theta(r) = M \text{ or } \Sigma_\Theta(r) = {\perp} \text{ and } \Sigma_\Delta(r) = (M,k)
	\end{align*}
	\caption{Auxiliary functions used by the formal semantics of \libOTM.}
	\label{fig:auxfuns}
\end{figure}

\paragraph{Semantics}
The machine dynamics is defined by the two transition relations induced by the rules in \Cref{fig:semantics-term,fig:semantics-io,fig:semantics-tau,fig:semantics-trs-mgr}; auxiliary definitions are in \cref{fig:auxfuns}.

The first relation $M \to N$ is defined on terms only, and models pure computations (\cref{fig:semantics-term}).
In particular, \cref{rule:admn-eval} allows a term $M$ that is not a value to be evaluated by means of an
auxiliary (partial) function $\mathcal{V}[M]$ yielding the value $V$; the other rules define the semantics of the monadic \type{bind} and exception handling in a standard way.
It is interesting to notice the symmetry between \type{bind} and \type{catch} and how \textcode{retry} is treated as an exception by \cref{rule:admn-binde} and as a result value by \cref{rule:admn-catchv}.

Relation $\to$ is used to define the labelled transition relation $\ctmSt{\Sigma}{P} \xrightarrow{\beta} \ctmSt{\Sigma'}{P'}$ over states.
This relation is non deterministic, to model the fact that the scheduler can choose among various threads to execute next; therefore, several rules can apply to a given state according to different evaluation contexts:
\begin{alignat*}{3}
	\textit{Expression: }\quad  \mathbb{E}\Coloneqq & \ctxhole \mid \mathbb{E} \textcode{ >>= } M
	&\qquad 
	\textit{Plain process: }  \mathbb P_t \Coloneqq &\ctmthread{\mathbb E}  \parallel P
	\\
	\textit{Transaction: }  \mathbb T_{t,k} \Coloneqq & \ctmthreadtr{\mathbb E}{M}{t}{k}  \parallel P
	&
	\textit{Any process: }  \mathbb A_t \Coloneqq & \mathbb{P}_t \mid \mathbb{T}_{t,k}	
\end{alignat*}

Labels $\beta$ describe the kind of transition, and are defined as follows:
\[
\beta \Coloneqq \tau
	\mid new\langle k\rangle 
	\mid co\langle k\rangle
	\mid ab\langle k, t, M\rangle
	\mid \overline{ab}\langle k, t, M\rangle
	\mid\ ?c
	\mid\ !c
\]
where  $k\in\TrName, M\in\Term$ as usual.

Transitions labelled by $\tau$ represent \emph{internal} steps of transactions, i.e., which do not need any coordination: reduction of pure terms, memory operations and thread creation (see rules in \cref{fig:semantics-tau}).
Reading a location falls into two cases: \cref{rule:read-miss}
models the reading of an unclaimed location and its effect is
to record the claim in $\Delta$, while \cref{rule:read-hit}
models the reading of a claimed location and its effect is
to merge the transactions of the current thread with that claiming the
location. Writes behave similarly.
\Cref{rule:orfirst,rule:orsecond} describe the semantics of alternative sub-transactions: if the
first one \textcode{retry}-es the second is executed discarding any effect of the first.
\Cref{rule:fork} spawns a new thread for the current transaction; a term \type{fork} $M$ can appear inside \textcode{atomic}, 
thus allowing multi-threaded open transactions, but its use inside \type{isolated}
is prevented by the type system and by the shape of \cref{rule:isolated} as well.

The remaining labels describe state transitions concerning the life-cycle of
transactions: creation, commit, abort, and restart (see rules in \cref{fig:semantics-trs-mgr}). These operations require a
coordination among threads; for instance, an abort from a thread has to be
propagated to every thread participating to the same
transaction.  This is captured in the semantics by labelling the transition
with the operation and the name of the transaction involved;
this information is used to force synchronisation of
all participants of that transaction. 
To illustrate this mechanism, we describe the commit of a transaction $k$,
namely $\ctmSt{\Sigma}{P} \xrightarrow{co\langle k\rangle} \ctmSt{\Sigma'}{P'}$.
First, by means of \cref{rule:tr-multicast-context} we split $P$ into two 
subprocesses, one of which contains all threads participating in $k$ 
(those not in $k$ cannot do a transition whose label contains $k$).
Secondly, using recursively \cref{rule:tr-multicast-commit} we single
out every thread in $k$. Finally, we apply \cref{rule:tr-commit}
provided that every thread is ready to commit, i.e., it is of the form 
$\ctmthreadtr{\textcode{return\;$M$}}{N}{t}{k}$.

Aborting a transaction works similarly, but it based on vetoes instead of an unanimous vote.
Aborts are triggered by unhandled exceptions raised by some thread, but threads react to this situation in different ways:
\begin{itemize}
	\item 
		threads forked within the transaction, in the same tree of the thread raising the exception: 
		these threads are killed (and the root thread aborted) because their creation must be discarded, as for any transactional side-effect;
	\item 
		threads from different trees which joined the transaction after it was created, due to a merging: these threads just retry their transaction, since aborting would require them to handle exceptions raised by ``foreign'' threads.
\end{itemize}
Like Haskell \libSTM \cite{hmpm:ppopp2005}, aborts leak some effects namely any transactional variable created in the aborted transaction that also occurs in the aborting exception.

Note that there are no derivation rules for $\textcode{retry}$: its meaning is
to inform the scheduler that we have reached a state where the execution is
stuck; hence the machine has to re-execute the transaction from the beginning
(or backtracking from a suitable check-point), possibly following a different execution order.

\FloatBarrier

\section{Opacity}\label{sec:opacity}
In this section we validate the formal semantics of \libOTM by proving it satisfies the \emph{opacity} correctness criterion for transactional memory \cite{gk:ppopp08}.

The opacity correctness criterion is an extension of the classical \emph{serialisability property} for databases with the additional requirement that even non-committed transactions must access consistent states. Intuitively, this property ensures that
\begin{enumerate}
\item effects of any committed transaction appear performed at a single, indivisible point during the transaction lifetime;
\item updates of any aborted transaction cannot be seen by other transactions;
\item transactions always access consistent states of the system.
\end{enumerate}

In order to formally capture these intuitive requirements let us recall some notions from \cite{gk:ppopp08}.
A \emph{history} is a sequence of \texttt{read}, \texttt{write}, \texttt{commit}, and \texttt{abort} operations\footnote{The definition in \cite{gk:ppopp08} considers finer-grained events; in particular, \texttt{read} and \texttt{write} operations are formed by \texttt{request}, \texttt{execution}, and \texttt{response} events. However in \textit{loc.~cit.}~the authors restrict to histories where \texttt{request}-\texttt{execution}-\texttt{response} sequences are not interleaved, hence we can consider the simpler \texttt{read}/\texttt{write}s events in the first place.}
ordered according to the time at which they were issued (simultaneous events are arbitrarily ordered) and such that no operation can be issued by a transaction that has already performed a \texttt{commit} or an \texttt{abort}.
A transaction $k$ is said to be in a history $H$ if the latter contains at least one operation issued by $k$.
Histories that differ only for the relative position of operations in different transactions are considered \emph{equivalent}.
Any history $H$ defines a \emph{happens-before} partial order $\prec_H$ over transactions, where
$k \prec_H k'$ iff the transaction $k$ becomes committed or aborted in $H$ before $k'$ issues its first operation.
If $\prec_H$ is total then $H$ is called \emph{sequential}.
For a history $H$, let $\mathit{complete}(H)$ be the set of histories obtained by adding either a commit or an abort for every live transaction in $H$.

We can now recall Guerraoui-Kapałka's definition\footnote{%
	The original definition requires the history $H$ to be ``legal'',
	but this notion is relevant only in presence of non-transactional
	operations which \libOTM prevents by design.
} 
of opacity {\cite[Def.~1]{gk:ppopp08}.
\begin{definition}[Opacity]
	A history $H$ is said to be \emph{opaque} if there exists a sequential
	history $S$ equivalent to some history in $\mathit{complete}(H)$ such that
	$S$ preserves the happens-before order of $H$.
\end{definition}

As shown in \cite{gk:ppopp08}, opacity corresponds to the absence of mutual dependencies between live transactions, where a dependency is created whenever a transaction reads an information written by another or depends from its outcome.
\begin{definition}[{Opacity graph \cite[Sec.~5.4]{gk:ppopp08}}]
	For a history $H$ 
	let $\ll$ be a total order on the set $T$ of all transactions in $H$.
	An \emph{opacity graph} $OPG(H,\ll)$ is a bi-coloured directed graph on $T$ such that a vertex is \emph{red} if the corresponding transaction is either running or aborted, it is \emph{black} otherwise,
	and for all vertices $k,k'\in T$, there is a edge $k \longrightarrow k'$ if any of the following holds:
	\begin{enumerate}
		\item 
			$k'$ happens-before $k$ ($k' \prec_H k$);
		\item
			$k$ reads something written by $k'$;
		\item	
			$k'$ reads some location written by $k$ and $k' \ll k$;
		\item
			$k'$ is neither running nor aborted and there are a location
			$r$ and a transaction $k''$ such that $k' \ll k''$, $k'$ writes to $r$, and $k''$ reads $r$ from $k$.
	\end{enumerate}	
    The edge is \emph{red} if the second case applies, otherwise it is black.
	The graph is said to be \emph{well-formed} if all edges from red nodes in $OPG(H,\ll)$ are also red.
\end{definition}

Let $H$ be a history and let $k$ be a transaction appearing in it. A \texttt{read} operation by $k$ is said to be \emph{local} (to $k$) whenever the previous operation by $k$ on the same location was a \texttt{write}. A \texttt{write} operation by $k$ is said to be \emph{local} (to $k$) whenever the next operation by $k$ on the same location is a \texttt{write}.
We denote by $\mathit{nonlocal}(H)$ the longest sub-history of $H$ without any local operations. A history $H$ is said \emph{locally-consistent} if every local \texttt{read} is preceded by a \texttt{write} operation that writes the read value; it is said \emph{consistent} if, additionally, whenever some $k$ reads $v$ from $r$ in $\mathit{nonlocal}(H)$ then
some $k'$ writes $v$ to $r$ in $\mathit{nonlocal}(H)$.

\begin{theorem}[{\cite[Thm.~2]{gk:ppopp08}}]
	A history $H$ is opaque if and only if
	\begin{enumerate}
		\item
			$H$ is consistent and
		\item
			there exists a total order $\ll$ on the set of transactions
			in $H$ such that $OPG(\mathit{nonlocal}(H),\ll)$ is well-formed and acyclic.
	\end{enumerate}
\end{theorem}

In \cite{gk:ppopp08} transactions may encapsulate several threads but cannot be merged. Therefore, in order to study opacity of \libOTM we extend the set of operations considered in \emph{loc.~cit.} with explicit merges.
Let $k,k'$ be two running transactions in the given history; when they merge, they share their threads, locations, and effects. From this perspective, $k$ is commit-pending and depends from $k'$ and hence in the opacity graph, $k$ is a red node connected to $k'$ by a red edge. Hence, merges can be equivalently expressed at the history level by sequences like:\\[1ex]
\begin{enumerate*}[label=\it(\arabic{*}),itemsep=-2pt]
	\item new $x$;
	\item $k'$ writes on $x$;
	\item $k$ reads from $x$;
	\item $k$ prepares to commit.
\end{enumerate*}\\[1ex]
These are the only dependencies found in histories generated by \libOTM.
\begin{theorem}\label{th:noloops}
	For $H$ a history describing an execution of a \libOTM program
	and a total order $\ll$, $OPG(\mathit{nonlocal}(H),\ll)$ is a forest of red
	edges where only roots may be black.
\end{theorem}
\begin{proof}
By inspection of the rules it is easy to see that
\begin{enumerate*}[label=\it(\alph{*})]
\item
transactions may access only locations they claimed;
\item
claimed locations are released only on \texttt{commit}s, \texttt{abort}s and retries;
\item
transactions have to merge with any transaction holding a location they need.
\end{enumerate*}
Therefore, at any given time there is at most one running transaction issuing operations on a given location, hence~\texttt{read}s and \texttt{write}s do not create edges.
Thus edges are created only during the execution of merges and, by inspecting the above implementation, it easy to see that
\begin{enumerate*}[label=\it(\alph{*}), resume]
\item
any transaction can issue at most one merge;
\item
a transaction issuing a merge is a red node;
\item
the edge created by a merge is red.
\end{enumerate*}
Therefore, transactions form a forest made of red edges where any non-root node is red.
\qed\end{proof}

Since  a forest formed by red edges whose sources are always red is always acyclic and well-formed, we can conclude our correctness result:
\begin{corollary}[Opacity]
	\libOTM meets the opacity criterion.
\end{corollary}

\section{Conclusions and future work}
\label{sec:conclusions}
In this paper we have presented OTM, a programming model supporting  interactions between composable memory transactions. This model separates isolated transactions from non-isolated ones, still guaranteeing atomicity; the latter can interact by accessing to shared variables. Consistency is ensured by transparently \emph{merging} interacting transactions at runtime.  
We have given a formal semantics for OTM, and proved that this model satisfies the important \emph{opacity} criterion.  

As future work, it would be interesting to add some heuristics to better handle \type{retry} events.  
Currently, a \type{retry} restarts all threads participating to the transaction; a more efficient implementation would keep track of the \emph{working set} of each thread, and at a \type{retry} we need to restart only the threads whose working sets have non-empty intersection with that being restarted. 
Another optimization is to implement transactions and OTVars directly in the runtime, akin the implementation of \libSTM in the Glasgow Haskell Compiler \cite{hmpm:ppopp2005}. 

We have presented OTM within Haskell (especially to leverage its type system), but this model is general and can be applied to other languages. A possible future work is to port this model to an imperative object oriented language, such as Java or C++; however, like other TM implementations, we expect that this extension will require some changes in the compiler and/or the runtime. 

This work builds on the calculus with shared memory and open transactions described in \cite{mpt:coord15}. In \textit{loc.~cit.}~this model is shown to be expressive enough to represent  $TCCS^m$ \cite{ksh:fossacs2014}, a variant of the Calculus of Communicating Systems with transactional synchronization.
Being based on CCS, communication in $TCCS^m$ is synchronous;
however, nowadays asynchronous models play an important r\^ole (see
actors, event-driven programming, etc.), so it may be interesting
to generalize the discussion so as to consider also this case, e.g.~by
defining an a calculus for event-driven models or an actor-based calculus with open transactions. 
Such a calculus can be quite useful also for modelling speculative reasoning for cooperating systems \cite{ma2010:speculative,mmp:eceast2014,mpm:gcm14w,mmp:dais14} or study distributed interacting transactions in serverles-computing \cite{jpbg:19,odcpm:20,gglmpx:19}.
A local version of actor-based open transactions can be
implemented in \libOTM using lock-free data
structures (e.g., message queues) in shared transactional memory.


\begin{thebibliography}{21}
\providecommand{\natexlab}[1]{#1}
\providecommand{\url}[1]{\texttt{#1}}
\expandafter\ifx\csname urlstyle\endcsname\relax
  \providecommand{\doi}[1]{doi: #1}\else
  \providecommand{\doi}{doi: \begingroup \urlstyle{rm}\Url}\fi

\bibitem[Donnelly and Fluet(2008)]{df:jfp08}
K.~Donnelly and M.~Fluet.
\newblock Transactional events.
\newblock \emph{J. Funct. Program.}, 18\penalty0 (5-6):\penalty0 649--706,
  2008.

\bibitem[Gabbrielli et~al.(2019)Gabbrielli, Giallorenzo, Lanese, Montesi,
  Peressotti, and Zingaro]{gglmpx:19}
M.~Gabbrielli, S.~Giallorenzo, I.~Lanese, F.~Montesi, M.~Peressotti, and S.~P.
  Zingaro.
\newblock No more, no less - {A} formal model for serverless computing.
\newblock In \emph{{COORDINATION}}, volume 11533 of \emph{Lecture Notes in
  Computer Science}, pages 148--157. Springer, 2019.

\bibitem[Guerraoui and Kapalka(2008)]{gk:ppopp08}
R.~Guerraoui and M.~Kapalka.
\newblock On the correctness of transactional memory.
\newblock In \emph{Proc.~PPOPP}, PPoPP '08, pages 175--184, New York,
  NY, USA, 2008. ACM.

\bibitem[Harris et~al.(2005)Harris, Marlow, Peyton~Jones, and
  Herlihy]{hmpm:ppopp2005}
T.~Harris, S.~Marlow, S.~L. Peyton~Jones, and M.~Herlihy.
\newblock Composable memory transactions.
\newblock In \emph{Proc.~PPOPP}, pages 48--60, 2005.

\bibitem[Herlihy and Moss(1993)]{moss:transactionalmemorybook}
M.~Herlihy and J.~E.~B. Moss.
\newblock Transactional memory: Architectural support for lock-free data
  structures.
\newblock In A.~J. Smith, editor, \emph{Proceedings of the 20th Annual
  International Symposium on Computer Architecture. San Diego, CA, May 1993},
  pages 289--300. {ACM}, 1993.

\bibitem[Jangda et~al.(2019)Jangda, Pinckney, Brun, and Guha]{jpbg:19}
A.~Jangda, D.~Pinckney, Y.~Brun, and A.~Guha.
\newblock Formal foundations of serverless computing.
\newblock \emph{Proc.~{OOPSLA}}:\penalty0
  149:1--149:26, 2019.

\bibitem[Koutavas et~al.(2014)Koutavas, Spaccasassi, and
  Hennessy]{ksh:fossacs2014}
V.~Koutavas, C.~Spaccasassi, and M.~Hennessy.
\newblock Bisimulations for communicating transactions - (extended abstract).
\newblock In A.~Muscholl, editor, \emph{Proc.~{FOSSACS}}, volume 8412
  of \emph{Lecture Notes in Computer Science}, pages 320--334. Springer, 2014.

\bibitem[Lesani and Palsberg(2011)]{lp:ppopp2011}
M.~Lesani and J.~Palsberg.
\newblock Communicating memory transactions.
\newblock In C.~Cascaval and P.~Yew, editors, \emph{Proc.~{PPOPP}}, pages 157--168.
  {ACM}, 2011.

\bibitem[Luchangco and Marathe(2011)]{lv:ppopp11}
V.~Luchangco and V.~J. Marathe.
\newblock Transaction communicators: Enabling cooperation among concurrent
  transactions.
\newblock In \emph{Proc.~PPOPP}, pages 169--178, 2011. ACM.

\bibitem[Ma et~al.(2010)Ma, Broda, Goebel, Hosobe, Russo, and
  Satoh]{ma2010:speculative}
J.~Ma, K.~Broda, R.~Goebel, H.~Hosobe, A.~Russo, and K.~Satoh.
\newblock Speculative abductive reasoning for hierarchical agent systems.
\newblock In J.~Dix, J.~Leite, G.~Governatori, and W.~Jamroga, editors,
  \emph{Computational Logic in Multi-Agent Systems}, volume 6245 of
  \emph{Lecture Notes in Computer Science}, pages 49--64. Springer Berlin
  Heidelberg, 2010.

\bibitem[Mansutti et~al.(2014{\natexlab{a}})Mansutti, Miculan, and
  Peressotti]{mmp:dais14}
A.~Mansutti, M.~Miculan, and M.~Peressotti.
\newblock Multi-agent systems design and prototyping with bigraphical reactive
  systems.
\newblock In K.~Magoutis and P.~Pietzuch, editors, \emph{Proc.~DAIS}, volume
  8460 of \emph{Lecture Notes in Computer Science}, pages 201--208. Springer,
  2014{\natexlab{a}}.

\bibitem[Mansutti et~al.(2014{\natexlab{b}})Mansutti, Miculan, and
  Peressotti]{mmp:eceast2014}
A.~Mansutti, M.~Miculan, and M.~Peressotti.
\newblock Distributed execution of bigraphical reactive systems.
\newblock \emph{{ECEASST}}, 71, 2014{\natexlab{b}}.

\bibitem[Mansutti et~al.(2014{\natexlab{c}})Mansutti, Miculan, and
  Peressotti]{mpm:gcm14w}
A.~Mansutti, M.~Miculan, and M.~Peressotti.
\newblock Towards distributed bigraphical reactive systems.
\newblock In R.~Echahed, A.~Habel, and M.~Mosbah, editors, \emph{Proc.~GCM'14},
  page~45, 2014{\natexlab{c}}.
\newblock Workshop version.

\bibitem[Miculan et~al.(2015)Miculan, Peressotti, and Toneguzzo]{mpt:coord15}
M.~Miculan, M.~Peressotti, and A.~Toneguzzo.
\newblock Open transactions on shared memory.
\newblock In \emph{{COORDINATION}}, volume 9037 of \emph{Lecture Notes in
  Computer Science}, pages 213--229. Springer, 2015.

\bibitem[Ni et~al.(2007)Ni, Menon, Adl{-}Tabatabai, Hosking, Hudson, Moss,
  Saha, and Shpeisman]{nietal:ppopp07}
Y.~Ni, V.~Menon, A.~Adl{-}Tabatabai, A.~L. Hosking, R.~L. Hudson, J.~E.~B.
  Moss, B.~Saha, and T.~Shpeisman.
\newblock Open nesting in software transactional memory.
\newblock In K.~A. Yelick and J.~M. Mellor{-}Crummey, editors,
  \emph{Proc.~{PPOPP}}, pages 68--78. {ACM}, 2007.

\bibitem[Obetz et~al.(2020)Obetz, Das, Castiglia, Patterson, and
  Milanova]{odcpm:20}
M.~Obetz, A.~Das, T.~Castiglia, S.~Patterson, and A.~Milanova.
\newblock Formalizing event-driven behavior of serverless applications.
\newblock In \emph{{ESOCC}}, volume 12054 of \emph{Lecture Notes in Computer
  Science}, pages 19--29. Springer, 2020.

\bibitem[Peyton~Jones(2001)]{jones:2010awkward-squad}
S.~L. Peyton~Jones.
\newblock Tackling the awkward squad: monadic input/output, concurrency,
  exceptions, and foreign-language calls in haskell.
\newblock \emph{Engineering theories of software construction}, 180:\penalty0
  47, 2001.

\bibitem[Peyton~Jones and Wadler(1993)]{pw:popl1993}
S.~L. Peyton~Jones and P.~Wadler.
\newblock Imperative functional programming.
\newblock In M.~S.~V. Deusen and B.~Lang, editors, \emph{Proc.~POPL}, pages
  71--84. {ACM} Press, 1993.

\bibitem[Peyton~Jones et~al.(1996)Peyton~Jones, Gordon, and
  Finne]{pgf:popl1996}
S.~L. Peyton~Jones, A.~D. Gordon, and S.~Finne.
\newblock Concurrent haskell.
\newblock In H.~Boehm and G.~L.~S. Jr., editors, \emph{Proc.~POPL}, pages 295--308. {ACM} Press, 1996.

\bibitem[Shavit and Touitou(1997)]{st:dc1997}
N.~Shavit and D.~Touitou.
\newblock Software transactional memory.
\newblock \emph{Distributed Computing}, 10\penalty0 (2):\penalty0 99--116,
  1997.

\bibitem[Smaragdakis et~al.(2007)Smaragdakis, Kay, Behrends, and
  Young]{skby:oopsla07}
Y.~Smaragdakis, A.~Kay, R.~Behrends, and M.~Young.
\newblock Transactions with isolation and cooperation.
\newblock In R.~P. Gabriel, D.~F. Bacon, C.~V. Lopes, and G.~L.~S. Jr.,
  editors, \emph{Proc.~{OOPSLA}}, pages 191--210. {ACM},
  2007.

\end{thebibliography}
\end{document}